\documentclass[runningheads,a4paper]{llncs}

\usepackage{amssymb}
\usepackage{graphicx}
\usepackage{amsmath}

\usepackage{url}
\urldef{\mailsa}\path|Rupei.Xu@utdallas.edu|
\urldef{\mailsb}\path|farago@utdallas.edu|
\urldef{\mailsc}\path|jjue@utdallas.edu|
\newcommand{\keywords}[1]{\par\addvspace\baselineskip
\noindent\keywordname\enspace\ignorespaces#1}

\begin{document}

\mainmatter  

\title{Job Edge-Fog Interconnection Network Creation Game in Internet of Things}

\titlerunning{Job Edge-Fog Interconnection Network Creation Game in Internet of Things}

%
%
\author{Rupei Xu%
\thanks{}%
\and Andr\'{a}s Farag\'{o} \and Jason P. Jue 
 }
\authorrunning{Rupei Xu \and Andr\'{a}s Farag\'{o} \and Jason P. Jue}

\institute{The University of Texas at Dallas\\800 W. Campbell Road\\
Richardson, TX 75080, USA\\
\mailsa\\
\mailsb\\
\mailsc\\
}

%
%

\toctitle{Job Edge-Fog Interconnection Network Creation Game in Internet of Things}
\tocauthor{Rupei Xu, Andr\'{a}s Farag\'{o} and  Jason P. Jue}
\maketitle

\begin{abstract}
	
This is the first paper to address the topology structure of Job Edge-Fog interconnection network in the perspective of network creation game. A two level network creation game model is given, in which the first level is similar to the traditional network creation game with total length objective to other nodes. The second level adopts two types of cost functions, one is created based on the Jackson-Wolinsky type of distance based utility, another is created based on the Network-Only Cost in the IoT literature. We show the performance of this two level game (Price of Anarchy). This work discloses how the selfish strategies of each individual device can influence the global topology structure of the job edge-fog interconnection network and provides theoretical foundations of the IoT infrastructure construction. A significant advantage of this framework is that it can avoid solving the traditional expensive and impractical quadratic assignment problem, which was the typical framework to study this task. Furthermore, it can control the systematic performance based only on one or two cost parameters of the job edge-fog networks, independently and in a distributed way.

\end{abstract}

\section{Introduction}

By the year 2020, major technology companies expect that the number of connected devices will be in the range of 25-50 billion. In particular, Cisco and Ericsson believe that 50 billion devices will be reached by 2020 \cite{ERIC} \cite{CISCO}. The Gartner Group on the other hand expects that number to be around 26 billion \cite{Gartner}. Internet of Things (IoT) typically involves a large number
of smart sensors sensing information from the environment and sharing it to a cloud service for processing. To tackle network issues involved in IoT and similar application computation, researchers have proposed bringing the computing cloud closer to data generators and consumers. One proposal is Fog computing cloud \cite{CISCO1} which lets network devices run cloud application logic on their native architecture. One may refer to the comprehensive survey on  edge and fog computing \cite{jj}. Mohan and Kangasharju \cite{EFC} introduced a distributed Edge-Fog Cloud framework for IoT computation. However, from the algorithmic point of view, it highly depends on the quadratic assignment problem framework, which is very impractical. In this paper, we give a new distributed framework which can avoid the  quadratic assignment problem. 

\subsection{Related Work}
\subsubsection{Network Creation Game}
Fabrikant et al. \cite{FA} formulated a well-studied game-theoretic model of network creation. Given $n$ agents (players), each one is corresponding to a vertex. The network is created by laying down connections (egdes) between vertices. The strategy of each agent $v$ is to choose a subset of the vertices $S_v$ to be connected with. In this formulation, each edge may appear twice, if $v$ lays a connection to $\omega$ and $\omega$ also lays a connection to $v$. Let a non-negative parameter $\alpha$ be the cost of making a connection. In this model, each agent desires to be close to other agents, besides spending little for buying links. Thus the total cost of each agent is defined as: 

$$cost(v)=\alpha\left|S_v \right|+ \sum\limits_{\omega} dist (v, \omega) ,$$

\noindent
where the sum is for all vertices in the created network and $dist (v, \omega)$ is the hop distance between vertices $v$ and  $\omega$, i.e. the number of edges on the shortest path between the two vertices. If there is no path between $v$ and $\omega$, the hop distance is infinity. Once a connection (link) is made in the network, all agents can use it regardless who paid the construction cost. From the cost function, we can know that on the one hand, each agent should pay some construction cost if he wants to connect to other agents, on the other hand, he prefers to be close to other nodes in the network.

\subsubsection{Quadratic Assignment Problem}

In 1957, Koopmans and Beckmann first introduced the  quadratic assignment problem (QAP) as a mathematical model for the locating a set of indivisible economical activities \cite{KB}. Later, Lawler \cite{LA} introduced a more general version of the QAP, in which a four-dimensional array of coefficients were given instead of the two matrices version in the Koopmans and Beckmann paper. It is known that QAP is one of the most difficult optimization problems. If $n>20,$ no 
algorithm is known with reasonable practical computation time for finding an exact solution. 
With $30$ nodes, even applying the advanced Kuhn-Munkres solver available from QAPLIB \cite{LIB}, it would take more than one week to get the exact solution\cite{EFC}. Sahni and Gonzalez \cite{SG} proved that QAP is NP-hard. Furthermore, it is also inapproximable, in the sense that it is impossible to find an $\epsilon$-approximate solution in polynomial time, unless P=NP. Those results hold even when the Koopamans-Beckmann coefficient matrices satisfy the triangle inequality \cite{TI}. The \emph{dense linear arrangement problem}, which is a special case of the Koopmans-Beckmann QAP, a polynomial time approximation scheme(PTAS) was found by Arora, Frieze and Kaplan \cite{ARORA}. There are also local search, SDP relaxations, ADMM, machine learning, sparsity as well as Gauss-Seidel decomposition based method etc. to study the QAP \cite{LOCAL}\cite{ADMM}\cite{NN}\cite{SPARSE}\cite{GS}.The reader is referred to the survey papers \cite{RE} \cite{S} for more information.

\section{Preliminaries}

In this paper, we propose a two level game model: edge-fog network creation game and job edge-fog interconnection network creation game. In the edge-fog network level, we adopt the traditional network creation game: SumGame; in the job edge-fog interconnection network level, we play an assignment game, i.e., each job chooses fog and edge devices to connect, to minimize its own cost. We consider the game-theoretic formation of interconnections between and within two networks:

Edge-Fog network $G_1=(V_1, E_1)$ and Job network $G_2=(V_2, E_2).$  The Game within the edge-fog network has $n_1$ players $\{1,...,n_1\}$, this set is denoted by $[n_1].$ The strategy space of each player is the set $S_i=2^{[n_1]-\{i\}}.$ The job edge-fog interconnection game has $n_2$ players $\{1,...,n_2\}$ in the set $[n_2].$ The strategy space of each player in job networks is the set $S_j=2^{[n_1]}.$ Let $\gamma(G)$ be the size of Minimum Dominating Set. The total interconnection edges between $G_1$ and $G_2$ is $I:=\cup_{j} S_{v_j}.$ In this paper we assume $n_1=n_2.$

The cost function for edge-fog  player $v_i$ is defined to be 

$$c(v_i)=\alpha\left|S_{v_i} \right|+ \sum\limits_{\omega\in G_1(V)} dist (v_i, \omega) .$$

The cost function for each job device $v_j$ is given  in two different types of games: 

\textbf{(Type I Game)}

$$c_1(v_j)=\beta|S_{v_j}|-\frac{1}{\sum\limits_{\omega\in G_1(V)} dist (v_j, \omega)} $$ 

\textbf{(Type II Game)}

$$c_2(v_j)=\beta|S_{v_j}|+\sum\limits_{\omega\in G_1(V)} dist (v_j, \omega)$$ 

Jackson and Wolinsky introduced \cite{JACK} a canonical problem in network formation which involves distance based utilities. Based on this, Shahrivar and Sundaram \cite{SS} introduced the interconnection network creation game. The Type I cost function in this paper is more specific. 

In this paper, in the Type II game, we consider about the Network-Only Cost (NOC) \cite{NOC} for each job device. Assume that each job device can connect to more than one edge-fog device,  and each edge-fog device can process several job devices. We also assume that the number of edge-fog devices and the number of job devices are the same. In case there are more job devices than edge-fog devices, we can split the existing edge-fog devices into virtual devices so that their number becomes equal to the number of jobs. Otherwise, the superfluous devices can be ignored.

\section{Best Response Strategies}

\begin{definition}\cite{GR} \cite{FD}
	Best response is the strategy (or strategies) which produces the most favorable outcome for a player, taking other players' strategies as given. 
\end{definition}

\begin{theorem} \cite{FA}
	It is NP-hard for each edge-fog player to find the best response strategy.
\end{theorem}

\begin{theorem}
	It is NP-hard for each job player to find the best response strategy (for both game types). 
\end{theorem}

\begin{proof}
	A dominating set for a graph $G=(V,E)$ is a subset $D$ of $V$ such that every vertex not in $D$ is adjacent to at least one member of $D$.	It is well known that the Minimum Dominating Set problem is NP-complete(it is among Karp's original 21 NP-Complete problems \cite{KARP}). It is easy to show that the best response strategy is in NP as it is verifiable in polynomial time for a given strategy. Next, we show a polynomial time many to one reduction from the Minimum Dominating Set problem to the best response strategy of job player. When $\beta>1,$ buying an edge is more expensive than the edge distance(as the graph is unweighted, every graph edge has length $1$), thus each job player prefers to connect to minimum number of vertices of $G_1$ with a relatively small total distance to all vertices of $G_1.$ Thus the best strategy would be to connect smallest number of vertices of $G_1$ to make sure that the distance from this job player to other vertices of $G_1$ which are not directly connected to is $2.$ Making more connections would only increase the total distance. Hence the cost is minimized when the vertices of $G_1$ which the job player connects to form the Minimum Dominating Set of $G_1.$
	
\end{proof}

\section{Nash Equilibrium and Price of Anarchy}

In Game Theory, the most beautiful fundamental concept is Nash Equilibrium, invented  by Nobel Prize laureate John Forbes Nash Jr.\cite{NASH}. 

\begin{definition}
	\textbf{Nash Equilibrium} is under which no agent can reduce its cost by unilaterally changing its strategy, if  others remain in the same strategy.  
	
\end{definition}

Nash also proved the existence of Nash Equilibrium in $n$ player game with mixed strategies in the same paper \cite{NASH}, using the Kakutani fixed-point theorem (Note: The Kakutani fixed-point theorem  is a min-max type result; a variant was also obtained by the Brouwer fixed-point theorem in the 1951 paper of Nash \cite{NASH1}).  It is known that computing the Nash Equilibrium belongs to the complexity class PPAD (Polynomial Parity Arguments on Directed graphs) \cite{CHEN}\cite{DAS}, which belongs to the TFNP class defined by Christos Papadimitriou \cite{PAPA} -- the complexity class of function problems that always guarantee the existence of the solution for NP search problems. 

In graph or network based game, we can also define the {\em Network Structure Equilibrium}, as follows: 

\begin{definition}
	\textbf{Network Structure Equilibrium} is the stable state of the network in which  no player has the incentive to change its connections. 
\end{definition}

The Price of Anarchy (PoA) \cite{WORST} is a concept in economics and game theory that measures how the efficiency of a system degrades due to selfish behavior of its agents.

\begin{definition}
	The \textbf{Price of Anarchy (PoA)} is the ratio of the maximum social cost incurred by any Nash Equilibrium and the minimum possible social cost incurred by any tuple of strategies. 
	
\end{definition}

The Nash equilibrium for edge-fog network is well studied in the literature. One can find the state of the art list of exiting results in \cite{sa}.  In this paper, we mainly focus on the second level game.

The social welfare of edge-fog network is the total cost of both edge and fog devices in this network. Its social welfare function is as follows: 

$$c(G_1)=\sum_{i}c_i=\alpha|E|+\sum_{i,j}d_{G_1}(i,j)$$

Based on the two types of utility functions, the following two cost functions of the job network can be obtained. 
\\
\textbf{(Type I Game)}
$$c_1(G_2)=\sum_{j} c_j=\beta\sum_{j} |S_{v_j}|-\sum_{j}\frac{1}{\sum_{\omega\in G_1}dist(v_j,\omega)}$$ 
\\
\textbf{(Type II Game)}
$$c_2(G_2)=\sum_{j} c_j=\beta\sum_{j} |S_{v_j}|+\sum_{j} \sum_{\omega\in G_1} dist(v_j,\omega)$$ 

Based on the fact that every pair of vertices not connected to each other by an edge is at least distance 2 away from each, \cite{FA} obtains a lower bound for $c(G_1).$

$$c(G_1)\geq \alpha|E|+2|E|+2(n(n-1)-2|E|)$$
$$=2n(n-1)+(\alpha-2)|E|.$$

Similarly, we can get the obivious lower bound for $c_1(G_2)$ and $c_2(G_2).$

\subsection{Type I Game}

\begin{lemma} (Reverse Cauchy-Schwarz Ineqaulity)
There exists a constant $c>0$, for  $0<a_i<U, i=1,...,n,$  $$\sum_{i=1}^{n}\frac{1}{a_i}\leq \frac{cU^2n^2}{\sum_{i=1}^{n} a_i}.$$
\end{lemma}

\begin{theorem} 
	The lower bound function of social cost for job network for type I game is $\beta|I|-\frac{4cn^4}{2n^2-|I|}.$
\end{theorem}

\begin{proof}
	
	Based on the fact that every pair of vertices not connected to each other between $G_1$ and $G_2$ are at least $2$ away from each other. The job players buy $|I|$ edges with a total cost $\beta|I|.$ Each vertex of $G_1$ connect to $G_2$ is $1$ away from the corresponding vertex of $G_2,$ other vertices of $G_1$ not connect to $G_2$, are at least $2$ away from the corresponding vertices of $G_2,$ thus this part the distance related cost is $-\sum_{j}\frac{1}{|S_{v_j}|+2(n-|S_{v_j}|)}.$ Put the two parts together, one can get the following inequality: 
	
	\begin{equation}
	\begin{aligned}	
	c_1(G_2)&\geq \beta|I|-\sum_{j}\frac{1}{|S_{v_j}|+2(n-|S_{v_j}|)}\\
	&=\beta|I|-\sum_{j}\frac{1}{(2n-|S_{v_j}|)}\\
	&\geq \beta|I|-\frac{c(2n)^2n^2}{\sum_{j}(2n-|S_{v_j}|)}\\
    &=\beta|I|-\frac{4cn^4}{2n^2-|I|}.
    \end{aligned}
    \end{equation}

\end{proof}

In this type of game, $c_1(G_2)$ is an unimodal function, which has a unique minimum value.  Take the deravitive, one can get $$c'_1(G_2)=\beta-\frac{4cn^4}{(2n^2-|I|)^2}.$$
Let $c'_1(G_2)=0,$ then $|I|=2n^2(1-\sqrt{\frac{c}{\beta}}).$

The social optimum is obtained in the saddle point. Thus $$c^*_1(G_2)=2\beta n^2(1-\sqrt{\frac{c}{\beta}})-\frac{4cn^4}{2n^2-2n^2(1-\sqrt{\frac{c}{\beta}})}=2n^2(\beta-2\sqrt{c\beta}).$$

\begin{theorem}
	When $0<\beta \leq 1,$ the job network PoA in Type I game is at most $\frac{1}{2-4\sqrt{\frac{c}{\beta}}}.$
\end{theorem}

\begin{proof}

When $0<\beta \leq 1,$ the Nash Equilibrium is a complete graph between the vertices of $G_1$ and $G_2.$ As the cost of buying one edge is less than the edge distance $1,$ players would buy the most to make their total cost minimum. 

The corresponding PoA is as follows:

$$PoA=\frac{\beta n^2-1}{2n^2(\beta-2\sqrt{c\beta})}\leq \frac{\beta n^2}{2n^2(\beta-2\sqrt{c\beta})}=\frac{\beta}{2(\beta-2\sqrt{c\beta})}=\frac{1}{2-4\sqrt{\frac{c}{\beta}}}.$$
\end{proof}

\begin{theorem}
	When $\beta>1,$ the job network PoA in Type I game is at least $$\frac{\gamma(G_1)}{2n(1-2\sqrt{\frac{c}{\beta}})}.$$
\end{theorem}

\begin{proof}
When $\beta>1,$ in Nash Equilibrium, the corresponding vertices of $G_1$ connected to $I$ is the minimum dominating set of $G_1.$ The cost of each player in Nash Equilibrium is 

$$c(v_j)\geq \beta\gamma(G_1)-\frac{1}{\gamma(G_1)+2(n-\gamma(G_1))}$$

Thus the corresponding PoA is: 

$$PoA\geq \frac{\beta\gamma(G_1)n-\frac{n}{2n-\gamma(G_1)}}{2n^2(\beta-2\sqrt{c\beta})}\geq \frac{\beta\gamma(G_1)}{2n(\beta-2\sqrt{c\beta})}=\frac{\gamma(G_1)}{2n(1-2\sqrt{\frac{c}{\beta}})}.$$

\end{proof}

\subsection{Type II Game}

\begin{theorem}
	The lower bound function of social cost for job network in type II game is $2n^2+(\beta-1)|I|.$
\end{theorem}

\begin{proof}
	Based on the fact that every pair of vertices not connected to each other between $G_1$ and $G_2$ are at least $2$ away from each other. The job players buy $|I|$ edges with a total cost $\beta|I|.$ Each vertex of $G_2$ connect to $G_1$ is $1$ away from the connected vertices of $G_2,$ this part the total distance is $|I|.$ Other vertices in $G_1$ not connect to $G_2$, are at least $2$ away from the corresponding vertices in $G_2,$ thus this part the total ditance is  $2(n^2-|I|).$ Add the three parts together, one can get the following inequality: 
	
	$$c_2(G_2)\geq \beta|I|+|I|+2(n^2-|I|)$$
	$$= 2n^2+(\beta-1)|I|.$$
	
\end{proof}

\begin{theorem}
	When $0<\beta \leq1,$ the job network PoA in Type II game is $1.$	
\end{theorem}

\begin{proof}
When $0<\beta \leq1,$ the social optimum is when $|I|$ is maximized, i.e., every vertex of $G_2$ is connected to every vertex of $G_1.$ Thus the social optimum is $c^*_2(G_2)=2n^2+(\beta-1)n^2=(\beta+1)n^2.$ The Nash Equilibrium in this case is represented by a complete graph between the vertices of $G_1$ and $G_2.$ As buying one connection edge is cheaper than the edge distance $1$, one would like to buy the most to make the total cost least. 

The corresponding price of anarchy is 

$$PoA=\frac{(\beta n+n)n}{(\beta+1)n^2}=1.$$ 
\end{proof}

\begin{theorem}
	When $1<\beta\leq 2,$ the job network PoA in Type II game is $1.$
\end{theorem}

\begin{proof}
When $1<\beta\leq 2$ the social optimum is obtained when $|I|$ is minimized, i.e., the vertices of $G_1$ that connected to $G_2$ vertices form a minimum dominating set of $G_1.$ 

The worst Nash Equilibrium is also obtained when $|I|$ is minimized, as the same with the social optimum. 

Because the cost of each edge is larger than the edge distance cost $1,$ buying one connection edge only decreases the distance of one pair vertices by $1$ if the original distance between them is $2,$ but add $\beta>1$ cost to the total.

The corresponding price of anarchy is $1$: 

$$PoA=\frac{2n^2+\gamma(G_1)(\beta-1)}{2n^2+\gamma(G_1)(\beta-1)}=1 .$$
\end{proof}

\begin{theorem}
	When $S< \beta\leq S+1,$ where $S\geq 3$, the job network PoA in Type II game is at most $\frac{S}{2}+1.$
\end{theorem}

\begin{proof}
If the original distance of one pair of vertices between $G_1$ and $G_2$ is greater or equal to $S\geq 3,$ buying one connection edge cost $ \beta\geq 2,$ the distance would decreased by $S-1\geq 2,$ then the players would like to buy edges to make sure each pair of vertices between  $G_1$ and $G_2$ is less than or equal to $2.$ Otherwise if $\beta> S-1,$ then in the Nash Equilibrium, players would buy smallest number of edges to make sure the distance of each pair of vertices between  $G_1$ and $G_2$ is less than or equal to $S.$

The corresponding price of anarchy is: 

\begin{equation}
\begin{aligned}
PoA&\leq \frac{\beta|I|+|I|+S(n^2-|I|)}{2n^2+\gamma(G_1)(\beta-1)}\\
&=\frac{(\beta+1-S)|I|+Sn^2}{2n^2+\gamma(G_1)(\beta-1)}\\
&\leq \frac{(\beta+1-S)n^2+Sn^2}{2n^2+\gamma(G_1)(\beta-1)}\\
&=\frac{(\beta+1)n^2}{2n^2+\gamma(G_1)(\beta-1)}\\
&\leq \frac{(\beta+1)}{2}\\
&\leq \frac{S}{2}+1.
\end{aligned}	
\end{equation}

\end{proof}

\section{Conclusion and Open Problems}

This is the first paper to address the topology structure of Job Edge-Fog interconnection network from the perspective of network creation game. A two level network creation game framework was developed. The analysis results show that, cleverly control the cost parameters of each player would lead the worst Nash Equilibrium very also or even exactly equal to the social optimum. This research opens a door to study the complicated job edge-fog interconnection network in an efficient distributed way. 

There are some open problems left that are worth mentioning: 

$\left(1\right)$  In this model, we consider the total length (sum of path lengths) in the cost function. This reflects an average case view.  If the total length is replaced by the maximum length, to reflect a worst case view, what can be shown about the performance?

$\left(2\right)$  In the social welfare functions, the interconnection cost has overlaps with the internal connection structure of edge-fog network. One may give a more accurate function to measure the social welfare, i.e., how the relationship of $\alpha$ and $\beta$ influences the two level network creation game performance?

$\left(3\right)$  In this framework, we only consider the network with distance based utility functions. As a refinement, one may also consider job process related issues in edge-fog devices. Then the processing time and job distribution would be important factors to analyze the performance of the job edge-fog interconnection network.

$\left(4\right)$  One may also develop other models rather than game theoretic framework to study the topology structure of job edge-fog interconnection networks. 

\paragraph{Acknowledgement}
The first author would like to thank Dr. Lin Chen for his inspiring discussions of computational difficulties of the QAP and Dr. Yi Li for his helpful discussions of the reverse version of Cauchy-Schwarz inequality.

\end{document}